\documentclass[english,aps, onecolumn]{revtex4}

\usepackage{graphicx}
\usepackage{dcolumn}
\usepackage[active]{srcltx}
\usepackage{color}
\usepackage{comment}
\usepackage{babel}
\usepackage{amsmath}
\usepackage{amstext}
\usepackage{amssymb}
\usepackage{amsthm}
\usepackage{amsfonts}
\usepackage{enumerate,enumitem}
\usepackage{lipsum}
\usepackage{ragged2e}
\usepackage{hyperref}
\hypersetup{
    colorlinks=true,
    citecolor = [rgb]{0.3,0.8,0},
    linkcolor=blue,
    filecolor=magenta,      
    urlcolor=[rgb]{0.4,0.7,0},
    breaklinks=true
}
\usepackage{breakurl}
\newcommand{\HS}[1]{{\color{red}Harshank: [#1]}}

\makeatletter
\@ifundefined{textcolor}{}
{%
 \definecolor{BLACK}{gray}{0}
 \definecolor{WHITE}{gray}{1}
 \definecolor{RED}{rgb}{1,0,0}
 \definecolor{GREEN}{rgb}{0,1,0}
 \definecolor{BLUE}{rgb}{0,0,1}
 \definecolor{CYAN}{cmyk}{1,0,0,0}
 \definecolor{MAGENTA}{cmyk}{0,1,0,0}
 \definecolor{YELLOW}{cmyk}{0,0,1,0}
}

\theoremstyle{plain}
\newtheorem{theorem}{Theorem}[section]
\newtheorem{lemma}{Lemma}[section]
\newtheorem{rem}[lemma]{Remark}
\newtheorem{corollary}[lemma]{Corollary}
\makeatother

\begin{document}
\begin{center}{\Large \textbf{
Certifying Temporal Correlations
}}\end{center}

\begin{center}
Harshank Shrotriya\textsuperscript{1*},
Leong-Chuan Kwek\textsuperscript{1,2},
Kishor Bharti\textsuperscript{3,4,5,6}
\end{center}

\begin{center}
{\bf 1} Centre for Quantum Technologies, National University of Singapore, 3 Science Drive 2, Singapore 117543
\\
{\bf 2} MajuLab, CNRS-UNS-NUS-NTU International Joint Research Unit, UMI 3654, Singapore
\\
{\bf 3} Joint Center for Quantum Information and Computer Science and Joint Quantum Institute, NIST/University of Maryland, College Park, Maryland 20742, USA
\\
{\bf 4} A*STAR Quantum Innovation Centre (Q.InC), Institute of High Performance Computing (IHPC), Agency for Science, Technology and Research (A*STAR), 1 Fusionopolis Way, \#16-16 Connexis, Singapore, 138632, Republic of Singapore
\\

{\bf 5}  Centre for Quantum Engineering, Research and Education, TCG CREST, Kolkata 700091, India
\\

{\bf 6} Science, Mathematics and Technology Cluster, Singapore University of Technology and Design, 8 Somapah Road, Singapore 487372, Singapore
\\
* harshank.s@u.nus.edu
\end{center}

\begin{center}
\today
\end{center}

\section*{Abstract}
{\bf
Self-testing has been established as a major approach for quantum device certification based on experimental statistics with minimal assumptions. However, despite more than 20 years of research effort, most of the self-testing protocols are restricted to spatial scenarios (Bell scenarios), without many temporal generalizations known. Under the scenario of sequential measurements performed on a single quantum system, semi-definite optimization based techniques have been applied to bound sequential measurement inequalities. Building upon this formalism, we show that the optimizer matrix that saturates such sequential inequalities is unique and, moreover, this uniqueness is robust to small deviations from the quantum bound. Furthermore, we consider a generalized scenario in the presence of quantum channels and highlight analogies to the structure of Bell and sequential inequalities using the pseudo-density matrix formalism. These analogies allow us to show a practical use of maximal violations of sequential inequalities in the form of certification of quantum channels up to isometries.
}

\vspace{10pt}
\noindent\rule{\textwidth}{1pt}
\noindent\rule{\textwidth}{1pt}
\vspace{10pt}

\section{Introduction}
Rapid development of quantum technologies has necessitated the need for certification of the crucial building blocks. Such building blocks could be logical qubits in a quantum computer or entangled states and measurement schemes in a quantum communication network. Among such certification tasks, quantum state and process certification hold an important place since most of such applications require preparation of a resourceful quantum state and subsequent manipulation via quantum processes. Techniques for state certification vary in the amount of assumptions required and information gained \cite{eisert2020quantum}, full tomographic reconstruction is the most resource costly of such certification schemes while providing tight confidence bounds on the certified quantum state. On the other end of the spectrum exist techniques such as self testing that utilise properties unique to quantum correlations while requiring weaker assumptions.

Entanglement of a spatially separated bipartite quantum state allows for device independent (DI) (i.e. requiring weak assumptions) certification of the underlying state up to local degrees of freedom owing to the Bell-nonlocality of the quantum correlations obtained from such a state. Since its inception in \cite{mayers_yao}, self-testing \cite{supic_self} has become a prominent technique to certify quantum states and measurements in the DI scenario \cite{yang_navascues,parallel_selftesting,mckague_st_parallel,coladangelo2017all,sarkar_self} by exploiting the properties of quantum realizations that achieve maximal violations of Bell inequalities and their uniqueness up to local isometries. Their scope has been extended to semi-device independent \cite{supic_epr_steering,gheorghiu2017rigidity,goswami_SDI,RST_steering,sarkar2021self}
and single party contextuality based scenarios \cite{kishor_self_test,contextuality_sumofsquares,hu2022self}. Applications of self testing in spatial scenarios have been utilised for device-independent (DI) randomness generation \cite{random_certified,random_2,random_3,random_4,random_5,random_6}, DI quantum cryptography via quantum key distribution \cite{ekert2014ultimate,mayers1998_crypto,crypto_1,crypto_2} and delegated quantum computing \cite{gheorghiu2019verification,RUV_classical,mckague2013interactive} among others.

Quantum correlations can also arise in a temporal setting where correlations are considered between measurement outcomes performed in a sequence \cite{brukner2004quantum, costa2018unifying, kull2019spacetime, temporal_correlations_geometry,fritz2010quantum}. In an analogy with spatial correlations, sequential measurement scenario also allows for the interpretation of such correlations as temporal correlations since measurements are done on the same state at different times. If commuting, such sequential measurements form a context and hence can be used to construct non-contextuality inequalities \cite{kochen1975problem}. However, in general the constraint of commuting sequential measurements can be lifted and such sequential correlations can be used to construct another class of inequalities such as the Leggett-Garg inequality \cite{leggett1985quantum}. For the particular case of Leggett-Garg inequality, the classical bound is obeyed by correlations that obey the postulates of macro-realism. Authors in \cite{otfried_seq} proposed a semi-definite programming (SDP) based method to find quantum bounds of sequential correlation-based inequalities. Temporal correlations, as obtained in the sequential measurement scenarios, have yet to find applications in quantum technologies save in few cases such as \cite{temporal_dim_witness} where authors show that genuine temporal correlations can be used to certify the minimum dimension of the underlying quantum state. In \cite{mohan2019sequential} authors consider the prepare-transform-measure scenario to certify a single qubit preparation (state) and intermediate measurements up to global unitary degree of freedom based on sequential correlation based witnesses. More recently, in \cite{das2022temporal}, a sequential measurement scenario has been considered where maximal violations of temporal inequalities were used to certify arbitrary outcome measurements under certain assumptions regarding the underlying state. An extension of sequential measurements to Bell scenarios is considered in \cite{bowles2020bounding} where an infinite hierarchy of SDPs is shown to bound the set of sequential correlations.

A unified framework to study spatio-temporal correlations for measurement events described across space and time has been proposed in \cite{PDO_first} where the authors introduce the Pseudo-density matrix (PDM) formalism. Using this approach, in \cite{temporal_teleport}, it was shown how time evolution of a qubit can be viewed as temporal teleportation in analogy to the spatial teleportation protocol that utilizes an entangled state as a resource \cite{spatial_teleport}. Further, in \cite{temporal_correlations_geometry}, the PDM formalism was utilized to characterize the geometry of spatio-temporal correlations arising from Pauli measurements of a qubit.

In this work, we study correlations arising from sequential measurements on a single party system and its generalizations to scenarios with quantum channels. Since such correlations cannot be used to certify the underlying state or measurements due to the single system isometries involved, we identify the correlation matrix itself as the candidate for certification. We start by reviewing in Section \ref{sec:prereq} the known results and essential concepts upon which we build this work. In Section \ref{sec:uniqueness} we utilise the SDP-based formulation to show that maximal violations of $N$-cycle sequential inequality leads to properties such as uniqueness of the optimizer matrix and also show the robustness of this uniqueness with small deviations from the maximal violation of the inequality. Next in Section \ref{sec:general_chn}, we consider a generalized scenario where quantum channels act between the sequence of measurements. Interestingly, it turns out that the SDP based optimization methods do not apply to this generalized scenario and to circumvent this hurdle we use the PDM formalism to put spatial (Bell) and temporal (sequential) inequality violations on the same footing. The PDM based approach enables us to formulate and prove results about the underlying channels along with highlighting some deeper connections between spatial and temporal correlation based scenarios. In Section \ref{sec:physical}, we present some physical implications of our main results along with applications. Finally, we conclude with a discussion of our results and future outlook in Section \ref{sec:conclusion}. 

\section{Pre-requisites}
\label{sec:prereq}
In this section, we briefly review some relevant concepts that would be utilized later to elucidate our results.

\subsection{Pseudo Density Matrices}
\label{PDM_appendix}
The density matrix of a system contains information of all possible Pauli observables acting on the system thus the density matrix has information of all correlations between spatially separated subsystems. The Pseudo-Density Matrix (PDM) formalism was introduced in \cite{PDO_first} by extending this analogy to account for causal correlations between observables acting on the same subsystem at different time points. That is to say, a PDM could be associated with any $n$-measurement event with there being a one-to-one correspondence between the $n$-measurement correlation values and the related PDM. The PDM of an $n$-measurement event is defined as
\begin{equation}
R=\frac{1}{2^{n}}\sum_{i_{1}=0}^{3}\cdots\sum_{i_{n}=0}^{3}\langle\{\sigma_{i_{j}}\}_{j=1}^{n}\rangle\bigotimes_{j=1}^{n}\sigma_{i_{j}}
\end{equation}
where $\sigma_{0}=\mathrm{I}$ and $\sigma_i$ for $i \in \{1,2,3\}$ are the familiar Pauli matrices. Further, $\{\ldots\}$ denotes a set of operators associated with the $n$-measurement event and not to be confused with anti-commutator. The sub-indices $j$ go over all measurement
events which could be done in a sequence on a single qubit, on spatially
separated qubits or a combination of both thus accounting for both spatial and causal correlations arising due to the $n$ measurement events. The factor $\langle\{\sigma_{i_{j}}\}_{j=1}^{n}\rangle$
denotes the correlation term corresponding to the size-$n$ measurement events. Physically, it is the expectation value of the product of the $n$ Pauli observables where tensor structure is appropriately imposed based on the spatial location of the measurement event. It should be noted that regardless of whether the measurements are performed at spatially separated qubits or not, the tensor structure is enforced on the operator $\bigotimes_{j=1}^{n}\sigma_{i_{j}}$ attached to the correlation factor $\langle\{\sigma_{i_{j}}\}_{j=1}^{n}\rangle$. This way of defining PDMs gives them their unique features that we briefly describe below. The properties retained by the PDM $R$ are hermiticity and
unit trace. If all the measurement events are performed on distinct qubits then $R$ is positive semi-definite (PSD) and thus a valid density matrix, however; in the presence of causal or temporal correlations
$R$ is not necessarily PSD with possible negative eigenvalues. As a measure of causality, authors in \cite{PDO_first} introduced $f_{tr}(R)$=$\left\Vert R\right\Vert _{tr}-1$. Let us expound on PDMs using the following example,\\
Example: Consider a qubit system, initially in state $|0\rangle$, on which 2-measurement events $E_1$ and $E_2$ are made in a sequence. Next, on calculating $\langle \{ \sigma_{i_1}, \sigma_{i_2} \} \rangle$ associated with the 2 events, we observe that only the following terms are non-zero,
\begin{align}
\langle \{ \sigma_1, \sigma_1 \} \rangle = 1, \quad \langle \{ \sigma_2, \sigma_2 \} \rangle = 1, \quad & \langle \{ \sigma_3, \sigma_3 \} \rangle = 1, \nonumber \\
\langle \{ \mathbb{I}, \mathbb{I} \} \rangle = 1, \quad \langle \{ \sigma_3, \mathbb{I} \} \rangle = 1, \quad & \langle \{ \mathbb{I}, \sigma_3 \} \rangle = 1.
\end{align}
Intuitively it is clear that since the two measurement events are performed at the same qubit in a sequence, the set $\{ \sigma_{i}, \sigma_{i} \}$ would have perfect correlation giving $\langle \{ \sigma_1, \sigma_1 \} \rangle = 1 \quad \forall \, i$. Subsequently, the PDM $R_{ex}$ for this sequence of 2 events on a single qubit system is
\begin{align*}
    R_{ex} =& \frac{1}{4}\sum_{i_1} \sum_{i_2} \langle \{ \sigma_{i_1}, \sigma_{i_2} \} \rangle (\sigma_{i_1} \otimes \sigma_{i_2}) \nonumber \\
    =& \begin{bmatrix}
        1 & 0 & 0 & 0 \\
        0 & 0 & 1/2 & 0 \\
        0 & 1/2 & 0 & 0 \\
        0 & 0 & 0 & 0 \\
        \end{bmatrix} 
\end{align*}
with eigenvalues $0,1,\frac{1}{2}, -\frac{1}{2}$. We see that $R_{ex}$ has a negative eigenvalue since the 2 measurement events $E_1$ and $E_2$ are causally related by virtue of acting on the same qubit system.

\subsection{Semidefinite programming basics}\label{SDP}
An $n \times n$ matrix $X$ is said to be positive semidefinite (PSD) if and only if $v^T X v \geq 0$ for any $v\in \mathbb{R}^n$, denote PSD matrix $X$ by $X \succeq 0$. It can be easily checked that the set of PSD matrices forms a convex cone. Then, a semidefinite program (SDP) is an optimization problem of the form
\begin{align}
    \mathrm{maximize:}\quad & \mathrm{Tr}(CX) \nonumber\\
    \mathrm{subject\,to}:\quad & \mathrm{Tr}(A_i X) = b_i \quad ,i=1,2,\ldots, m \\
                            & X\succeq 0,   \nonumber    
\end{align}
where we note that the objective function $\mathrm{Tr}(CX)$ is linear in $X$ and the $m$ linear equation constraints are given by $m$ matrices $A_1,\ldots, A_m$ and the $m$-vector $b$. If the optimal value for $\text{Tr}(CX)$ exists and is finite then it is called the primal optimal value ($p^*$) and it is attained at the primal optimal solution ($X^*$).
Referring to the problem above as the primal (P) SDP, we can define the dual of it as given below.\\
\textit{Semidefinite programming duality}: The dual (D) problem (SDD) of the above SDP is defined to be 
\begin{align}
    \mathrm{minimize:}\quad & y^{T}b = \sum_{i=1}^m y_i b_i \nonumber\\
    \mathrm{subject\,to}:\quad & \sum_{i=1}^m y_i A_i - C = S \\
                            & S\succeq 0,   \nonumber    
\end{align}
If the optimal value for $\sum_{i=1}^m y_i b_i$ exists and is finite then it is called the dual optimal value ($d^*$) and is attained at the dual optimal solution ($y^*$). Next we summarize some relevant results from SDP duality theory.
\begin{theorem}
Consider a pair of primal (P) and dual (D) SDPs. The following holds: \\

\begin{itemize}
    \item (Complementary slackness) Let $X,(y, S)$ be a pair of primal-dual feasible solutions for (P) and (D), respectively. Assuming that $p^* = d^*$ we have that $X,(y, S)$ are primal-dual optimal if and only if $\langle X, S \rangle = 0$. \\
    \item (Strong duality) Assume that $d^* > -\infty$ (resp. $p^* < +\infty$) and that (D) (resp. (P)) is strictly feasible. Then $p^* = d^* $ and furthermore, the primal (resp. dual) optimal value is attained.
\end{itemize}
\end{theorem}

\textit{Dual nondegeneracy:} Let $Z^*$ be an optimal dual solution and let $M$ be any symmetric matrix. If the homogeneous linear system 
\begin{equation}
MZ^* = 0,
\end{equation}

\begin{equation}
\textrm{Tr}(MA_i) = 0 \ (\forall i \in [m]),
\end{equation}
only admits the trivial solution $M=0$, then $Z^*$ is said to be dual nondegenerate.

\subsection{Bounding temporal correlations via SDP}
\label{subsec_otfried}
Consider the scenario of sequential measurements performed on a single qudit (of dimension $d$) state. Each of the measurement observable $\{A_{i}\}$ gives a binary value outcome labeled by $\left(\pm1\right)$. Quantum bounds on linear expressions constructed from correlations between sequential measurements of binary valued  observables can be obtained using a semi-definite programming (SDP) based approached as shown in \cite{otfried_seq}. For an expression $C=\sum_{ij}\lambda_{ij}X_{ij}$ where $\lambda_{ij}$ are the coefficients corresponding to the term $X_{ij}=\langle A_{i}A_{j}\rangle_{seq}$ with

\begin{equation} \label{eq:seq_cor}
    \langle A_{i}A_{j}\rangle_{seq}=\frac{1}{2}\left[Tr\left(\rho A_{i}A_{j}\right)+Tr\left(\rho A_{j}A_{i}\right)\right]
\end{equation} the optimization problem is given by;

\begin{align}
\mathrm{maximize:}\, & \sum_{ij}\lambda_{ij}X_{ij},\\
\mathrm{subject\,to}:\, & X=X^{T}\succeq0\;\mathrm{and\,}X_{ii}=1\;\forall\,i \label{eq:seq_sdp}
\end{align}

The constraint $X\succeq0$ follows from the fact that $X$ is the real part of matrix $Y$ with $Y_{ij}=\mathrm{Tr}\left[\rho(A_{i}A_{j})\right]$ and $v^{T}Yv\geq0$
for any real vector $v$. The result of the optimization program is an optimizer matrix, denoted by $X^{opt}$, which achieves the maximum for $\sum_{ij}\lambda_{ij}X_{ij}$. Consider a special case of the expression above which we call $S_N$, given by
\begin{equation}
    S_{N} \equiv \sum_{i=1}^{N-1}\langle A_{i}A_{i+1}\rangle_{seq}-\langle A_{N}A_{1}\rangle_{seq} \label{eq:sn}
\end{equation}
 with $N \geq 3$. $S_N$ has a classical bound $N-2$ \cite{araujo2013all}, using which we can define an N-cycle inequality as $S_N \leq N-2$ which is obeyed by all macro-realistic theories. Invoking the strong duality for the corresponding SDP gives the quantum bound $S_{N}\leq N\cos\left(\frac{\pi}{N}\right)$ \cite{otfried_seq}. In other words, the primal optimal of the SDP being bounded implies that the dual SDP is also feasible and moreover the primal and the dual optima are equal (Strong Duality Theorem).

\begin{rem} \label{col_X}
 Reinterpreting the correlation term $X_{ij}$ as an inner product
of unit vectors $\{x_{i}\}$, we can write $X_{ij}=\left(x_{i},x_{j}\right)$
thus obtaining $\{x_{i}\}$ as the columns of the matrix $\sqrt{X}$.
As a consequence of this reinterpretation, for every positive semi-definite matrix $X$ one can find a set $\{x_{i}\}$ and a set
of binary ($\pm1$) outcome observables $\{A_{i}\}$ such that 
\begin{equation}
    \langle A_{i}A_{j}\rangle_{seq}=Tr\left[\frac{1}{2}\rho\left(A_{i}A_{j}+A_{j}A_{i}\right)\right]=\left(x_{i},x_{j}\right), \label{eq:unit_vecs}
\end{equation}
for \textit{all} quantum states $\rho$. Thus, for every quantum state $\rho$ there exist observables $\{A_{i}\}$ such that $X_{ij}^{opt}=\langle A_{i}A_{j}\rangle_{seq}=Tr\left[\frac{1}{2}\rho\left(A_{i}A_{j}+A_{j}A_{i}\right)\right]$ which maximally violate the N-cycle inequality $S_N \leq N-2$. Further, in the qubit case it can be seen that taking $A_i = \Vec{a_i}.\Vec{\sigma}$ gives $\langle A_{i}A_{j}\rangle_{seq} = \Vec{a_i}.\Vec{a_j}$ implying that the correlations do not depend on the underlying quantum state.
\end{rem}

One interesting aspect of this approach is the fact that in the temporal scenario a single SDP suffices to find a tight upper bound in the general case which is in contrast with the infinite (NPA) hierarchy of SDPs required in the general case to tightly bound Bell inequalities in the spatial correlation scenario as shown in \cite{npa_prl,npa_njp}. This significant increase in complexity can be attributed to the fact that the observables acting on the spatially separated halves of the entangled state are required to commute but no such restriction is placed on observables acting on a single quantum system at distinct time points in the sequential measurement scenario.

\section{Certification of temporal correlations using optimizer matrix \label{sec:uniqueness}}
In this section, we build on the methodology outlined in the previous section by proving results for the optimal correlation set that maximally violates the N-cycle inequality. Concretely, our first main result is regarding the uniqueness of this optimal set and stated as,
\begin{theorem} 
The optimizer matrix $X^{opt}$ that maximizes the objective function
$S_N$ in \eqref{eq:sn} is unique for all $N \geq 3$.
\label{thm:uniqueness}
\end{theorem}

In order to prove this result we need to formulate the dual program that exists for every primal semi-definite program and utilize relevant results from SDP duality theory. The dual of the SDP in ~\eqref{eq:seq_sdp} with the choice of objective function as $S_N$ \eqref{eq:sn} is given by
\begin{equation}
\min\text{ }\sum_{i}y_{i} \nonumber
\end{equation}
such that
\begin{equation}
\sum_{i=1}^{N} y_{i} e_{i}e_{i}^{T}-\Lambda \succcurlyeq 0,\label{eq:dual_sdp}
\end{equation}
where $e_{i}$ is an $N\times1$ column vector with $1$ at the $i^{\text{th}}$ place and $0$ elsewhere and \begin{equation*}
    \Lambda = -0.5\left(e_{1}e_{N}^{T}+e_{N}e_{1}^{T}\right)+\sum_{i=1}^{N-1}0.5\left(e_{i}e_{i+1}^{T}+e_{i+1}e_{i}^{T}\right).
\end{equation*}
Moreover, the primal and the dual are both feasible with the primal and dual optima being equal to $N \cos{\frac{\pi}{N}}$. Then let us first prove the following intermediary result,
\begin{lemma}  \label{dual_optimal}
The dual optimal solution for the dual program in~\eqref{eq:dual_sdp} is given by

\begin{equation}
W_{N}=\cos\left(\frac{\pi}{N}\right)\mathbb{I}_{N}+0.5\left(e_{1}e_{N}^{T}+e_{N}e_{1}^{T}\right)-\sum_{i=1}^{N-1}0.5\left(e_{i}e_{i+1}^{T}+e_{i+1}e_{i}^{T}\right),\label{eq:dual_optimal}
\end{equation}

where $\mathbb{I}_{N}$ is an $N\times N$ identity matrix and $e_{i}$
is an $N\times1$ column vector with $1$ at the $i$th place and
$0$ elsewhere. 
\end{lemma}

\begin{proof}
Given the dual formulation in \eqref{eq:dual_sdp}, we start by claiming that the optimal choice of $\{y_i\}$ that achieves the dual optima $N \cos{\frac{\pi}{N}}$ is such that 
\begin{equation*}
    \sum_{i=1}^{N} y_{i} e_{i}e_{i}^{T}-\Lambda = W_N \implies y_i = \cos{\frac{\pi}{N}} \quad \forall i
\end{equation*}
On comparing coefficients we see that the choice above results in  $\sum_i y_{i} = N\cos\left(\frac{\pi}{N}\right)$ which agrees with the dual optimal value. It remains to be shown that $W_{N}$ is positive semidefinite for all $N \geq 3$. We proceed by rewriting $W_N$ as
\begin{equation}
    W_N = c_N \mathbb{I}_{N} + 0.5  \underbrace{\begin{bmatrix}
   0 & -1 & 0 & \cdots & \cdots & 1 \\
   -1 & 0 & -1 & \cdots & \cdots & 0 \\
   0 & -1 & 0 & -1 & \cdots & 0 \\
   \vdots  & \vdots & \vdots  & \ddots  & \vdots & \vdots  \\
   \vdots  & \vdots & \vdots  & \vdots & \ddots & \vdots  \\
   1 & 0 & \cdots & \cdots & -1 & 0 
 \end{bmatrix}}_{\textrm{call this matrix } T_N}
\end{equation}
where $c_N = \cos{\frac{\pi}{N}}$. Note that $W_N$ is PSD if all eigenvalues of $T_N$ are greater than or equal to $-2c_N$ which is what we show next. In order to find the eigenvalues of $T_N$, consider the determinant
\begin{equation}
    |T_N - \lambda \mathbb{I}_N| = \begin{vmatrix}
   -\lambda & -1 & 0 & \cdots & \cdots & 1 \\
   -1 & -\lambda & -1 & \cdots & \cdots & 0 \\
   0 & -1 & -\lambda & -1 & \cdots & 0 \\
   \vdots  & \vdots & \vdots  & \ddots  & \vdots & \vdots  \\
   \vdots  & \vdots & \vdots  & \vdots & \ddots & \vdots  \\
   1 & 0 & \cdots & \cdots & -1 & -\lambda 
 \end{vmatrix}
\end{equation}
The matrix above is an ordinary tridiagonal matrix and the determinant can be evaluated using the following formula \cite{molinari2008determinants} involving multiplication of $2 \times 2$ matrices 
\begin{align}
    \begin{vmatrix}
   a_1 & b_1 & \cdots & c_0 \\
   c_1 & a_2 & b_2 & \vdots \\
   \vdots & \ddots & \ddots & \vdots \\
   b_n & \cdots & c_{n-1} & a_n
 \end{vmatrix} = &(-1)^{n+1} (b_n \cdots b_1 + c_{n-1} \cdots c_0) +  \nonumber \\
  &\textrm{Tr}\left[ \begin{pmatrix} a_n & -b_{n-1}c_{n-1} \\ 1 & 0 
 \end{pmatrix} \cdots
 \begin{pmatrix} a_2 & -b_{1}c_{1} \\ 1 & 0 
 \end{pmatrix}
 \begin{pmatrix} a_1 & -b_{n}c_{0} \\ 1 & 0 
 \end{pmatrix} \right]
\end{align}
Plugging the values of matrix elements for our case gives 
\begin{equation}
    |T_N - \lambda \mathbb{I}_N| = (-1)^{N+1} (2(-1)^{N-1}) + \textrm{Tr} \left[
    \begin{pmatrix} -\lambda & -1 \\ 1 & 0 
 \end{pmatrix}^{N} \right]
\end{equation}
Note that the matrix $P \equiv \begin{pmatrix}
 -\lambda & -1 \\ 1 & 0
\end{pmatrix}$ has eigenvalues $\mu = \frac{-\lambda + \sqrt{\lambda^2-4}}{2}$ and $\frac{1}{\mu}$. Also, $P$
is not diagonalizable for $\lambda = 2$ which is an eigenvalue of $T_N$ for odd $N$. However, $W_N$ remains PSD in this case.\\
For $\lambda \neq 2$, we have 
\begin{equation}
    \textrm{Tr}\left[ \begin{pmatrix}
     -\lambda & -1 \\ 1 & 0
    \end{pmatrix}^N \right] = \mu^N + \frac{1}{\mu^N}
\end{equation}
giving 
\begin{equation}
    |T_N - \lambda \mathbb{I}_N| = 2 + \mu^N + \frac{1}{\mu^N}
\end{equation}
The expression above vanishes if $(\mu^N + 1)^2=0$. The roots of $\mu^N + 1 = 0$ are $\mu = e^{i(2m+1)\pi/N}$ with $m \in \mathbb{Z}$. Thus,
\begin{align}
    \frac{-\lambda + \sqrt{\lambda^2-4}}{2} &= e^{i(2m+1)\pi/N} \nonumber \\
    \sqrt{\lambda^2-4} &= \lambda + 2e^{i(2m+1)\pi/N} \nonumber \\ 
    \lambda^2 - 4 &= \lambda^2 + 4e^{i2(2m+1)\pi/N} + 4\lambda e^{i(2m+1)\pi/N} \nonumber \\
    \lambda &= -2 \cos{\frac{(2m+1)\pi}{N}} \geq -2\cos{\frac{\pi}{N}} \quad \forall \, m \textrm{ and all }N \geq 3
\end{align}
Therefore, $W_N$ is PSD for all $N \geq 3$ concluding the proof of Lemma \ref{dual_optimal}, that the choice $W_N$ leads to the dual optimal solution.
\end{proof}

\begin{lemma} \label{all_zero}
The only solution to the system of linear equations 
\[
X_{N}W_{N}=0
\]
 is $X_{N}=0$, where $X_{N}$ is a symmetric $N\times N$ matrix
with diagonal elements equal to zero and $W_{N}$ is the dual optimal solution given in \eqref{eq:dual_optimal}.
\end{lemma}

\begin{proof} The proof is trivial and follows from simple linear algebra.
There are $N^{2}$ linear equations where the maximum number of variables
in a equation is three. Let us consider the equations with variables
corresponding to the first row of $X_{N}$. There are $N$ such linear
equations. Furthermore, there are three equations with number of variables
equal to two. These three equations fix the value of the variables
involved either equal to each other (one such equation) or a constant
times the other variable. Here, the constant is $2\cos{\frac{\pi}{N}}$. Substituting these constraints in the equations
with three variables, we get new equations with two variables. Following
this approach, we get the value of all the variables corresponding row $1$ of $X_N$ equal to zero as the only self-consistent solution.
Same argument applies for the variables in the other rows. This completes
the proof.
\end{proof}

\begin{lemma} (taken from ~\cite{alizadeh}) \label{alizadeh}
Let $Z^*$ be a dual optimal and non-degenerate solution of a semi-definite program. Then, there exists a unique primal optimal solution for that SDP. 
\end{lemma}

Using the lemmas above, the proof of Theorem~\ref{thm:uniqueness} follows as

\begin{proof}
Lemmas ~\ref{dual_optimal} and \ref{all_zero} imply that the dual optimal $W_N$ is non-degenerate. Together with Lemma~\ref{alizadeh}, this implies that the primal optimal for the SDP in~\eqref{eq:seq_sdp} with objective function as in \eqref{eq:sn} is unique. This completes the proof of Theorem~\ref{thm:uniqueness}.
\end{proof}

We mentioned in Remark \ref{col_X} that the matrix $X$ that is being optimized over is PSD and gives us the set of unit vectors $\{x_i\}$ as the columns of $\sqrt{X}$; however, these $\{x_i\}$ vectors are not directly relatable to any state or measurements. Noting that $X$ is the real part of PSD matrix $Y$, where $Y_{ij} = \mathrm{Tr}(\rho A_i A_j)$, allows us to introduce set of unit vectors $\{y_i\}$ as columns of $\sqrt{Y}$ matrix. In the special case when $\rho=|\psi\rangle\langle\psi|$ is a pure state, we can write $y_i = A_i|\psi\rangle$. For $X^{opt}$, we thus have $\displaystyle X^{opt}_{ij}= \mathrm{Tr} \left[\frac{1}{2}\tilde{\rho}\left(\tilde{A_{i}}\tilde{A_{j}}+\tilde{A_{j}}\tilde{A_{i}}\right)\right] =\mathrm{Re\left\{ \mathrm{Tr} \left[\tilde{\rho}\tilde{A_{i}}\tilde{A_{j}}\right]\right\} } =\mathrm{Re} \left\{ \langle\tilde{\psi}|\tilde{A_{i}}\tilde{A_{j}}|\tilde{\psi}\rangle\right\} =\mathrm{Re} \left\{ \left(\tilde{y_{i}},\tilde{y_{j}}\right)\right\} 
$ where in the third step we have assumed that $\tilde{\rho}$ is a pure state, since $X^{opt}$ can be obtained using any quantum state and suitable 2-outcome measurements. However, note that while defining the set $\{\tilde{y_{i}}\}_{i}$
we have a global isometry such that there could be another set of unit vectors $\{y'_{i}\}_{i}$ satisfying $\mathrm{Re}\left\{ \left(\tilde{y_{i}},\tilde{y_{j}}\right)\right\} =\mathrm{Re}\left\{ \left(Uy'_{i},Uy'_{j}\right)\right\} $
with $UU^{\dagger}=I$. 
This leads to the following,

\begin{corollary}
For the optimizer matrix $X^{opt}$  which maximizes an N-cycle inequality
of the kind \eqref{eq:sn}, for every set $\{y_{i}\}_{i}$ which
satisfies $X^{opt}_{ij}=\mathrm{Re}\left\{ \left(y_{i},y_{j}\right)\right\} $,
we have a global isometry $U$ and a reference set $\{\tilde{y_{i}}\}_{i}$
such that 
\begin{equation}\label{eq7}
\mathrm{Re}\left\{ \left(\tilde{y_{i}},\tilde{y_{j}}\right)\right\} =\mathrm{Re}\left\{ \left(Uy_{i},Uy_{j}\right)\right\} 
\end{equation}
Furthermore this implies the existence of reference measurements $\tilde{\{A_{i}\}_{i}}$ and pure state $\tilde{|\psi\rangle}$ such that $\displaystyle
\tilde{y_{i}}=\tilde{A_{i}}|\tilde{\psi}\rangle $. 
\end{corollary}

Note that in Bell inequality based self testing scenarios uniqueness up to local isometries of $A_i |\psi\rangle$ implies self test of the measurement $A_i$; however, in the temporal scenario due to the single system global isometry involved we can only claim that uniqueness of $X^{opt}$ fixes $\mathrm{Re}\{(y_i,y_j)\}$.

Having established the uniqueness of the optimizer matrix, next we ask how robust is this uniqueness property associated with $X^{opt}$. Precisely, let us imagine a scenario where one obtains the set of correlations $ \{ \langle A_i A_j \rangle_{seq} \}$ from an experimental setup and constructs a candidate matrix $X$ s.t. $X_{ij} = \langle A_i A_j \rangle_{seq}$. From the set of sequential correlations obtained experimentally, the value of the expression associated with an objective function (such as \eqref{eq:sn}) can be calculated. Assuming experimental imperfections, one can inquire how a small deviation of the value thus calculated from the maximal value (obtained using $X^{opt}$) relates to the distance between the matrices $X$ and $X^{opt}$. We formally define this notion as,\\

\textit{$(\epsilon, r)$ robust certification of temporal correlations:} Given 2 outcome observables $\{A_i \}_i$ and physical set of sequential correlations $\{ \langle A_i, A_{j} \rangle_{seq} \}_i$, they give an $(\epsilon, r)$ certificate of reference temporal correlations $\{\langle \tilde{A}_i, \tilde{A}_{j} \rangle \}_{ij}$ if the matrices $X$ and $\tilde{X}$ defined via $X_{i,j} = \langle A_i A_j \rangle$ (similarly for $\tilde{X}$) are close in Frobenius norm distance s.t.
\begin{equation*}
    |X - \tilde{X}| \leq \mathcal{O}(\epsilon^{r}).
\end{equation*}
We will make use of the above definition as our metric for the proximity of two sets of sequential/temporal correlations in our certification result. Let us first state the following useful lemma,

\begin{lemma} (taken from ~\cite{kishor_self_test}) \label{robust_contextuality}
Consider a pair of primal/dual SDPs (P) and (D), where the primal/dual values are equal and both are attained. Furthermore, assume that the set of feasible solutions of (P) is contained in some compact subset $U \subseteq \mathcal{S}^n$. Let $\mathcal{P}$ be the set of primal optimal solutions and $d$ be the singularity degree of (P) defined (in \cite{borwein1981facial}) as the least number of facial reduction steps required to make (P) strictly feasible, we have that
\begin{equation*}
    \textrm{dist}(\tilde{X}, \mathcal{P}) \leq \mathcal{O}(\epsilon^{2^{-d}}),
\end{equation*}
for any primal feasible solution $\tilde{X}$ with $p^* - \epsilon \leq \langle C, \tilde{X} \rangle$.
\end{lemma}

Now we are ready to establish our result concerning the robust certification of temporal correlations,

\begin{theorem} 
    Robustness: Consider the SDP \eqref{eq:seq_sdp}, given that matrix $X^{real}$ achieves a near-optimal value for the objective function $C=\sum_{ij}\lambda_{ij}X_{ij}$ such that $|\sum_{ij}\lambda_{ij}X^{opt}_{ij} - \sum_{ij}\lambda_{ij}X^{real}_{ij}| \leq \epsilon$ and $X^{opt}$ is unique then it follows that the correlations making up $X^{real}$ give an $(\epsilon,1)$ certificate of reference correlations making up $X^{opt}$ i.e. we can upper bound the Frobenius-norm distance between $X^{opt}$ and $X^{real}$ as
    \begin{equation}
        |X^{opt} - X^{real}| \leq \mathcal{O}(\epsilon).
    \end{equation}
    \label{thm:robustness}
\end{theorem}

\begin{proof}
Since the identity matrix belongs to the strictly feasible set of the SDP in~\eqref{eq:seq_sdp}, the singularity degree is $0$. Once we substitute $d=0$ in the statement of Lemma~\ref{robust_contextuality}, we recover our robustness statement.
\end{proof}

The multiplicative constant going with $\mathcal{O}(\epsilon)$ certification guarantee in Theorem \ref{thm:robustness} can be approximated by numerical SWAP based methods. Such SWAP based techniques along with SDP based methods were used to estimate constants in CHSH and CGLMP inequality based robust self tests in \cite{bancal2015physical}. Also, analytical robust self testing bounds for certain Bell inequalities have been obtained in \cite{kaniewski2016analytic,li2019analytic}. The sequential measurement scenario considered here is closer in spirit to the scenario considered in \cite{hu2022self} where numerical SWAP based techniques were adapted for robust self testing of a single quantum system. 

\begin{corollary}
It follows from Theorem \ref{thm:uniqueness} and \ref{thm:robustness} that near maximal violation of the inequality $S_N \leq N-2$, where $S_N$ is the N-cycle expression from \eqref{eq:sn} with quantum bound $\beta_q = N \cos{\frac{\pi}{N}}$, by the set $\{ \langle A_i, A_{i+1} \rangle_{seq} \}_i$ such that 
\begin{equation*}
    \left| \left(\sum_{i=1}^{N-1}\langle A_{i}A_{i+1}\rangle_{seq}-\langle A_{N}A_{1}\rangle_{seq}\right) - \beta_q \right| \leq \epsilon
\end{equation*} is an $(\epsilon, 1)$ certificate for the reference set $\{\langle \tilde{A}_i, \tilde{A}_{j} \rangle_{seq} \}_{ij}$ with $\langle \tilde{A}_i, \tilde{A}_{j} \rangle_{seq} = X^{opt}_{i,j}$. 
\end{corollary}

\section{Generalized scenario with quantum channels}
\label{sec:general_chn}
In this section we generalize the sequential measurement scenario considered in the previous part by introducing a quantum channel between the measurements of Alice and Bob as shown in Fig. \ref{fig_channel}. A similar setup is studied in \cite{chen2024semi} by defining a framework in terms of instrument moment matrices (IMMs) which are completely different from our construction of sequential correlation matrix. Our construction is inspired by the semi-definite optimization based formulation in prior works such as \cite{otfried_seq}. In the case of arbitrary channels $\mathcal{E}_{A|B}$ acting in between the time points the sequential correlation cannot be calculated using expression \eqref{eq:seq_cor}. Consider the following sequential measurement protocol (see Fig.~\ref{fig_channel}),

\begin{itemize}
    \item Step 1: Alice obtains a quantum state given by $\rho_A$  \\
    \item Step 2: Alice performs projective 2 outcome measurement $A_m$ on $\rho_A$. Note that the choice of $A_m$ is restricted to the set of measurements $\{ M_i \}$. The post-measurement state $\rho_m^{x}$ with outcome $x =\pm 1$ is \\
    \begin{equation*}
        \rho_m^{\pm 1} = \Pi^{\pm}_m \rho_A \Pi^{\pm}_m 
    \end{equation*}
    \item Step 3: Post measurement state $\rho_m^x$ is transferred to Bob via the channel $\mathcal{E}_{A|B}$,
    \begin{equation*}
        \mathcal{E}_{A|B} (\rho^{\pm 1}_m) = \mathcal{E}_{A|B} (\Pi^{\pm}_m \rho_A \Pi^{\pm}_m) 
    \end{equation*}
    \item Step 4: Bob performs projective 2 outcome measurement $A_n \in \{ M_i \}$ on $\mathcal{E}_{A|B} (\rho^x_m)$. The resulting sequential correlations are given by,
    \begin{align*}
        \langle A_m A_n \rangle_{seq} &= P(++ \text{ or } --) - P(+- \text{ or } -+) \\
        &= P(+_B|+_A)P(+_A) + P(-_B|-_A)P(-_A) - P(+_B|-_A)P(-_A) - P(-_B|+_A)P(+_A) \\
        &= \text{Tr}\left[ \Pi^+_n \mathcal{E}_{A|B} (\Pi^+_m \rho_A \Pi^+_m) \right] + \text{Tr}\left[ \Pi^-_n \mathcal{E}_{A|B} (\Pi^-_m \rho_A \Pi^-_m) \right] \\ & \quad -\text{Tr}\left[ \Pi^+_n \mathcal{E}_{A|B} (\Pi^-_m \rho_A \Pi^-_m) \right] - \text{Tr}\left[ \Pi^-_n \mathcal{E}_{A|B} (\Pi^+_m \rho_A \Pi^+_m) \right] \\
        &= \text{Tr}[(\Pi^+_n - \Pi^-_n)\mathcal{E}_{A|B} (\Pi^+_m \rho_A \Pi^+_m)] - \text{Tr}[(\Pi^+_n - \Pi^-_n)\mathcal{E}_{A|B} (\Pi^-_m \rho_A \Pi^-_m)] \\
        &= \text{Tr}[A_n \mathcal{E}_{A|B} (\Pi^+_m \rho_A \Pi^+_m)] - \text{Tr}[A_n \mathcal{E}_{A|B} (\Pi^-_m \rho_A \Pi^-_m)] \\
        &= \text{Tr} [\mathcal{E}_{A|B}^{\dagger}(A_n) (\Pi^+_m \rho_A \Pi^+_m - \Pi^-_m \rho_A \Pi^-_m)] \\
        &= \text{Tr}\Big[ \mathcal{E}_{A|B}^{\dagger}(A_n) \frac{(A_m \rho_A + \rho_A A_m)}{2} \Big].
    \end{align*}
\end{itemize}
where in the last line we switch the action of channel from $\mathcal{E}_{A|B}$ acting on state to $\mathcal{E}_{A|B}^{\dagger}$ acting on observable $A_n$. Similar to the previous case, we can define the matrix $X$ such that
\begin{equation*}
    X_{mn} = \langle A_m A_n \rangle_{seq} = \frac{1}{2} \mathrm{Tr} \Big[ (A_m \rho_A + \rho_A A_m) \mathcal{E}_{A|B}^{\dagger}(A_n) \Big].
\end{equation*}

\begin{figure}[h]
  \includegraphics[width=0.75\columnwidth]{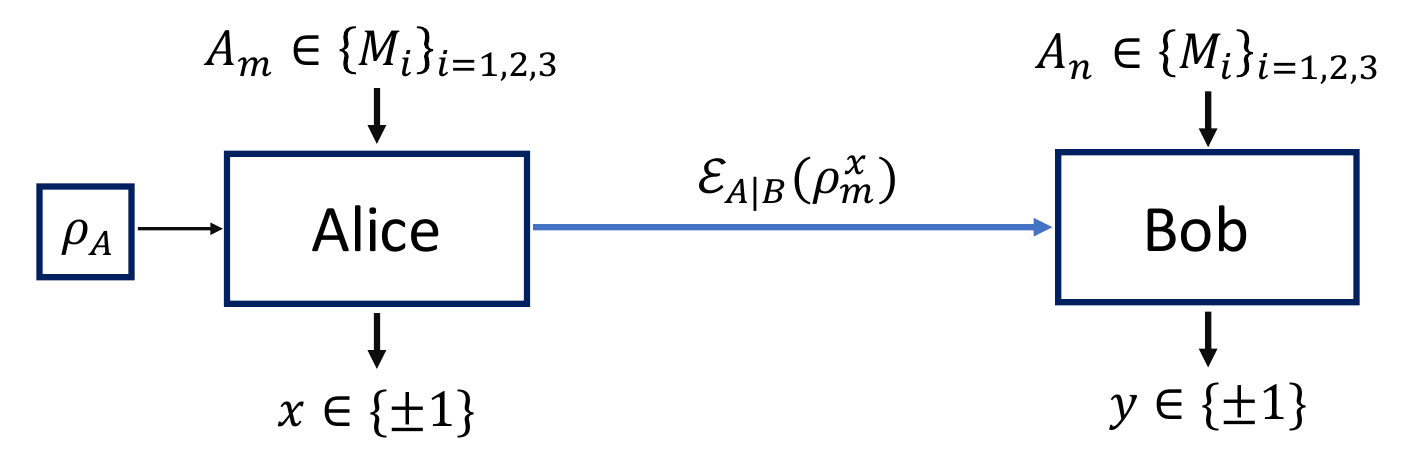}
  \caption{In the generalized scenario, Alice obtains black box containing $\rho_A$ from source and makes measurement $A_m \in \{M_i\}$ at time $t_A$, then the black box containing post measured state $\rho_m^x$ is sent via channel $\mathcal{E}_{A|B}$ to Bob who performs measurement $A_n \in \{M_i\}$ at time $t_B$.}
  \label{fig_channel}
\end{figure}

Unlike the previous case the matrix $X$ defined via $X_{mn} = \langle A_m A_n \rangle_{seq}$ is not generally PSD in this scenario. Hence, one cannot use the SDP method described earlier in Sec. \ref{subsec_otfried} to bound sequential measurement inequalities. The sequential measurement scenario can be viewed from the perspective of Pseudo Density Matrix (PDM) formalism (see section \ref{PDM_appendix} for a review) by considering the measurement events in a sequence and writing the corresponding PDM. The PDM formalism is more general since it can also take into account time evolution (via quantum channels) of the quantum state between the two sequential measurement events. We will make use of this feature of PDMs to make progress in the following section. Consider a 2-dimensional quantum state $\rho_A$ with sequential measurements being performed at time points $t_A$ and $t_B$, additionally the channel $\mathcal{E}_{A|B}$ acts on the post measured state at $A$ mapping operators from state space $\mathcal{H}_A$ at $t_A$ to state space $\mathcal{H}_B$ at $t_B$ (Fig.~\ref{fig_channel}). Then it can be shown that the PDM ($R_{AB}$) for this scenario is given by (see \cite{temporal_correlations_geometry} for a derivation)
\begin{align} \label{PDM_chn}
    R_{AB} &= (\mathcal{I}_A \otimes \mathcal{E}_{A|B}) \Big\{ \rho_A\otimes \frac{I}{2},\; \frac{1}{2} \sum_{i=0}^{3} \sigma_i \otimes \sigma_i \Big\}
\end{align}
where $\mathcal{I}_A$ is the identity superoperator and $\{A,B\} := AB + BA $. The utility of the PDM formalism will be made clear via its role in proving our main result in this scenario regarding channel certification (see Theorem \ref{channel_cert_thm}). To that end, let us first establish some useful lemmas.

\begin{lemma} \label{lemma_4}
In the generalized scenario described above, we can relate the sequential correlation to the relevant PDM as,
$$
\langle A_{m}A_{n}\rangle_{seq} = \mathrm{Tr}\left[(A_{m}\otimes A_{n})R_{AB}\right]
$$
\end{lemma}

\begin{proof} First let us define the Jamiolkowski isomorphic operator to $\mathcal{E}_{A|B}$ as
$$
E_{AB} = \sum_{ij}(\mathcal{I}_A \otimes \mathcal{E}_{A|B})(|i\rangle \langle j|_A \otimes |j\rangle \langle i|_B)
$$ which acts on  $\mathcal{H}_A \otimes \mathcal{H}_B$. Then starting with RHS, 

\begin{equation}
\begin{aligned}
\mathrm{Tr}\left[(A_{m}\otimes A_{n})R_{AB}\right] & =\mathrm{Tr}\left\{ (A_{m}\otimes A_{n})\left[\left(\rho_{A}\otimes\frac{\mathrm{I}}{2}\right)E_{AB}+E_{AB}\left(\rho_{A}\otimes\frac{\mathrm{I}}{2}\right)\right]\right\} \quad \text{using \eqref{PDM_chn}} \\
 & =\frac{1}{2}\left[\mathrm{Tr}\left(\left(A_{m}\rho\otimes A_{n}\right)E_{AB}\right)+\mathrm{Tr}\left(\left(A_{m}\otimes A_{n}\right)E_{AB}(\rho\otimes\mathrm{I})\right)\right] \\
 & =\frac{1}{2}\left[\sum_{ij}\mathrm{Tr}\left(A_{m}\rho|i\rangle\langle j|\right)\mathrm{Tr}\left(A_{n}\mathcal{E}_{A|B}\left(|j\rangle\langle i|\right)\right)+\sum_{ij}\mathrm{Tr}\left(A_{m}|i\rangle\langle j|\rho\right)\mathrm{Tr}\left(A_{n}\mathcal{E}_{A|B}\left(|j\rangle\langle i|\right)\right)\right]\\
 & =\frac{1}{2}\left[\sum_{ij}\langle j|(A_{m}\rho+\rho A_{m})|i\rangle\mathrm{Tr}\left(A_{n}\mathcal{E}_{A|B}\left(|j\rangle\langle i|\right)\right)\right] \\
 & =\frac{1}{2}\left[\sum_{ij}\langle j|(A_{m}\rho+\rho A_{m})|i\rangle\mathrm{Tr}\left(A_{n}\sum_{k}V_{k}|j\rangle\langle i|V_{k}^{\dagger}\right)\right] \text{ where $\{V_k\}$ are the Kraus ops. for $\mathcal{E}_{A|B}$} \\
 & =\frac{1}{2}\left[\sum_{ij}\langle j|(A_{m}\rho+\rho A_{m})|i\rangle\langle i|\sum_{k}V_{k}^{\dagger}A_{n}V_{k}|j\rangle\right] \\
 & =\frac{1}{2}\left[\sum_{j}\langle j|(A_{m}\rho+\rho A_{m})\sum_{k}V_{k}^{\dagger}A_{n}V_{k}|j\rangle\right] \\
 & =\frac{1}{2}\mathrm{Tr}\left[\sum_{j}|j\rangle\langle j|(A_{m}\rho+\rho A_{m})\sum_{k}V_{k}^{\dagger}A_{n}V_{k}\right]\\
 & =\frac{1}{2}\mathrm{Tr}\left[(A_{m}\rho+\rho A_{m})\mathcal{E}_{A|B}^{\dagger}(A_{n})\right] \\
 & =\langle A_{m}A_{n}\rangle_{seq}.
\end{aligned}    
\end{equation}
\end{proof}


 \textit{Sequential inequalities in the general scenario:} Let us consider the sequential inequalities from the previous section; however, being constructed out of sequential correlations of the generalised scenario with the quantum channel acting between the measurements. Having established a formal connection of the PDM formalism with sequential correlations in Lemma~\ref{lemma_4}, we can put it to use to explore the maximal violation of sequential inequalities in the generalized scenario and implication on the underlying quantum channel. Recall the $N$-cycle sequential correlation expression \eqref{eq:sn} and the associated inequality,
 $$
 S_N = \sum_{i=1}^{N-1}\langle A_{i}A_{i+1}\rangle_{seq}-\langle A_{N}A_{1}\rangle_{seq} \leq N-2.
 $$
 Using Lemma \ref{lemma_4} we can rewrite the above as
\begin{equation}\label{seq_ineq_PDM}
    \mathrm{Tr} \Bigg[ \Bigg(\sum_{i=1}^{N-1} A_{i}\otimes A_{i+1} - A_{N}\otimes A_{1} \Bigg)R_{AB}\Bigg] \leq N-2,
\end{equation}
where it is interesting to note that \eqref{seq_ineq_PDM} possesses a structure similar to Bell inequalities under entanglement-based scenarios where the tensor product appears due to spatially separated Hilbert spaces (see Section \ref{sec:physical}).

\textit{Channel certification using N-cycle sequential inequality:} Consider the case of N-cycle inequality described above with $N=3$ which takes the form of the well known Leggett-Garg inequality \cite{leggett1985quantum},
\begin{equation}\label{eq:3cycle}
    S_3 \equiv \langle A_1 A_2 \rangle_{seq} + \langle A_2 A_3 \rangle_{seq} - \langle A_3 A_1 \rangle_{seq} \leq 1.
\end{equation}
We will consider maximal violation of this inequality in the presence of channels. Throughout the draft we highlighted the various degrees of freedom involved in obtaining correlations from sequential measurements on a quantum system. We saw that in the simple case without channels being present, sequential correlations depend only on the angle between the observables and not on the particular direction of the observable and not even on the underlying quantum state. However, in the generalised scenario with channels other isometries also come into play thus allowing us to choose non ideal angles to obtain maximal violation of the 3-cycle sequential inequality \eqref{eq:3cycle}. Following this discussion, we do not base our self testing statement on certification of the quantum state or of the measurements but rather on the quantum channel acting in between the sequence of measurements. A protocol for channel certification has been previously described in \cite{channel2018certifying} where the physical channel was compared to a target reference channel via the action of the channel on the maximally entangled state thus requiring the physical input state to be an entangled one. On the contrary, our channel certification scheme does not require entanglement but utilizes sequential correlations instead.

\begin{theorem} \label{channel_cert_thm}
Temporal certification statement: In a sequential measurement setting with channel $\mathcal{E}_{A|B}$ acting between the sequence of 2-dim. measurements $\{A_i\}$ on 2-dim. state $\rho$, maximal violation of the 3-cycle sequential inequality \eqref{eq:3cycle} implies that $\mathcal{E}_{A|B}$ is a 1-qubit Pauli channel with Kraus rank 1.
\end{theorem}

\begin{proof}
We first prove the following useful result,
\begin{lemma} \label{lemma:pauli_corr}
Consider the correlation $\langle \sigma_k \sigma_l \rangle_{seq}$ in the generalized scenario with channels with $\sigma_{k/l}$ as the Pauli observables and $\rho_A$ being 2-dimensional. Then,
\begin{equation}
    \langle \sigma_k \sigma_l \rangle_{seq} = \frac{1}{2}\Big[ \langle \sigma_k \rangle_{\rho_A} \mathrm{Tr}(\sigma_l \mathcal{E}_{A|B}(\sigma_0)) + \mathrm{Tr}(\sigma_l \mathcal{E}_{A|B}(\sigma_k)) \Big]
\end{equation}
\end{lemma}

\begin{proof}
\begin{align}
    \langle\sigma_k \sigma_l \rangle_{seq} &= \mathrm{Tr} [(\sigma_k \otimes \sigma_l)R_{AB}] \quad \mathrm{using\;Lemma\;\ref{lemma_4}} \nonumber \\
            &= \mathrm{Tr} \left[(\sigma_k \otimes \sigma_l) (\mathcal{I}_A \otimes \mathcal{E}_{A|B}) \Big\{ \rho_A\otimes \frac{I}{2},\; \frac{1}{2} \sum_{i=0}^{3} \sigma_i \otimes \sigma_i \Big\} \right] \quad \mathrm{using\;\eqref{PDM_chn}\;\textrm{from MT}} \nonumber \\
            &= \frac{1}{4} \mathrm{Tr} \left[(\sigma_k \otimes \sigma_l) (\mathcal{I}_A \otimes \mathcal{E}_{A|B})  \sum_{i=0}^{3} \{\rho_A, \sigma_i\} \otimes \sigma_i \right] \nonumber \\
            &= \frac{1}{4} \mathrm{Tr} \left[\sum_{i=0}^{3} \sigma_k \{\rho_A, \sigma_i\} \otimes \sigma_l \mathcal{E}_{A|B}(\sigma_i) \right] \nonumber \\
            &= \frac{1}{4} \sum_{i=0}^{3} \mathrm{Tr}(\sigma_k \rho_A \sigma_i + \sigma_k \sigma_i \rho_A).\mathrm{Tr}(\sigma_l \mathcal{E}_{A|B}(\sigma_i)) \nonumber \\
            &= \frac{1}{4} \sum_{i=0}^{3} \mathrm{Tr} (\rho_A(\sigma_i \sigma_k + \sigma_k \sigma_i)).\mathrm{Tr}(\sigma_l \mathcal{E}_{A|B}(\sigma_i)) \nonumber \\
            &= \frac{1}{2} \left[ \langle \sigma_k \rangle_{\rho_A} \mathrm{Tr}(\sigma_l \mathcal{E}_{A|B}(\sigma_0)) + \sum_{i=1}^{3} \delta_{ik} \mathrm{Tr}(\sigma_l \mathcal{E}_{A|B}(\sigma_i)) \right] \nonumber \\
            &= \frac{1}{2}\Big[ \langle \sigma_k \rangle_{\rho_A} \mathrm{Tr}(\sigma_l \mathcal{E}_{A|B}(\sigma_0)) + \mathrm{Tr}(\sigma_l \mathcal{E}_{A|B}(\sigma_k)) \Big]. \nonumber
\end{align}
\end{proof}

Proceeding with the proof of Theorem \ref{channel_cert_thm}, take $\rho_A = |0\rangle \langle 0|$ since we expect maximal violation to be obtainable using a pure state and rotational freedom of observables allows us to take it as an eigenstate of $\sigma_z$ . Then using Lemma \ref{lemma:pauli_corr}, we get
\begin{equation}\label{eq:pauli_corr}
    \langle \sigma_k \sigma_l \rangle_{seq} =
\begin{cases}
    \mathrm{Tr}(\sigma_l \mathcal{E}_{A|B}(\sigma_k))/2, \text{ for } k=1,2\\
    \mathrm{Tr}(\sigma_l \mathcal{E}_{A|B}(\sigma_0) + \sigma_l \mathcal{E}_{A|B}(\sigma_3))/2, \text{ for } k=3
\end{cases}
\end{equation}

Furthermore, it has been established in \cite{channels_M2} that all possible quantum channels acting on 2-dim. states correspond to the convex closure
of the maps parametrized in the Pauli basis $\{\sigma_{i}\}_{i}$
using the following Kraus operators,
\begin{equation} \label{eq:chn_convex}
\begin{array}{cc}
K_{+} & =\left[\cos\frac{v}{2}\cos\frac{u}{2}\right]\sigma_{0}+\left[\sin\frac{v}{2}\sin\frac{u}{2}\right]\sigma_{3}\\
K_{-} & =\left[\sin\frac{v}{2}\cos\frac{u}{2}\right]\sigma_{1}-i\left[\cos\frac{v}{2}\sin\frac{u}{2}\right]\sigma_{2}
\end{array}
\end{equation}
with $v\in[0,\pi]$, $u\in[0,2\pi]$. 
Action of the channel $\mathcal{E}_{A|B}(\sigma_i)$ can be written using these Kraus operators as,
\begin{equation}\label{kraus_ops}
\begin{aligned}K_{+}\sigma_{0}K_{+}^{\dagger}+K_{-}\sigma_{0}K_{-}^{\dagger} & =\sigma_{0}+\sin(u)\sin(v)\sigma_{3}\\
K_{+}\sigma_{1}K_{+}^{\dagger}+K_{-}\sigma_{1}K_{-}^{\dagger} & =\cos(u)\sigma_{1}\\
K_{+}\sigma_{2}K_{+}^{\dagger}+K_{-}\sigma_{2}K_{-}^{\dagger} & =\cos(v)\sigma_{2}\\
K_{+}\sigma_{3}K_{+}^{\dagger}+K_{-}\sigma_{3}K_{-}^{\dagger} & =\cos(u)\cos(v)\sigma_{3}.\\
\end{aligned}
\end{equation}
Substituting \eqref{kraus_ops} into \eqref{eq:pauli_corr} gives
\begin{align}
\label{pauli_corr_angles}
    \langle \sigma_k \sigma_l \rangle_{seq} &= \delta_{kl} (\delta_{1k} \cos u + \delta_{2k} \cos v + \delta_{3k} \cos (u-v))
\end{align}
Using \eqref{pauli_corr_angles} for $S_3$ along with the fact that for 2-dimensional observables $A_{m/n} = \Vec{a}_{m/n}. \Vec{\sigma}$ we can write $\langle A_m A_n \rangle_{seq} = \sum_{k,l} a_{mk} a_{nl} \langle\sigma_k \sigma_l \rangle_{seq}$ gives,

\begin{equation}
\begin{aligned} \label{eq:S3}
 S_3 &= \sum_{k,l=1}^{3} \big[ (a_{1k}a_{2l} + a_{2k}a_{3l} - a_{3k}a_{1l})
    \delta_{kl}(\delta_{1k} \cos u + \delta_{2k} \cos v + \delta_{3k} \cos (u-v)) \big] \\
    &= \sum_{k=1}^{3} \big[ (a_{1k}a_{2k} + a_{2k}a_{3k} - a_{3k}a_{1k})
    (\delta_{1k} \cos u + \delta_{2k} \cos v + \delta_{3k} \cos (u-v)) \big] \\
\end{aligned}
\end{equation}

From \eqref{eq:S3}, it can be seen that the maximal value achieved is $3/2$ in the cases with $\{u=v=0\}$, $\{u=v=\pi\}$, $\{u=0, v=\pi \}$ and $\{u=\pi, v=0 \}$ which correspond to special cases of the Pauli channel. However, note that the angle between the measurement choices $\{ A_i \}$ would depend on the particular case of the channel $\mathcal{E}_{A|B}$. Hence, taking a convex combination of channels corresponding to the cases above will not give maximal violation for a particular choice of 3 measurements $\{ A_i \}_{i=1,2,3}$. This result again highlights the unitary degree of freedom in choosing the channel as described in \eqref{eq:local_iso}. Interestingly, the 4 channels corresponding to the cases above also correspond to the 4 ``pseudo-Bell states"
\begin{align}
    R_{AB}^{(1)} &= \frac{1}{4} (I + X\otimes X + Y \otimes Y + Z \otimes Z) \nonumber \\[0.75em]
    R_{AB}^{(2)} &= \frac{1}{4} (I + X\otimes X - Y \otimes Y - Z \otimes Z) \nonumber \\[0.75em]
    R_{AB}^{(3)} &= \frac{1}{4} (I - X\otimes X + Y \otimes Y - Z \otimes Z) \nonumber \\[0.75em]
    R_{AB}^{(4)} &= \frac{1}{4} (I - X\otimes X - Y \otimes Y + Z \otimes Z) \nonumber
\end{align}
which were used for showing temporal evolution as an analogue of spatial teleportation in \cite{temporal_teleport}.
\end{proof}

\section{Physical implications and Applications}
\label{sec:physical}
Next, we provide some physical implications of our main results followed by some potential applications of the main theorems. 

\textit{Unique set of correlations:} We establish in Theorems \ref{thm:uniqueness} and \ref{thm:robustness} that for the $N$-cycle inequality in the sequential measurement setup the set of correlations $\{A_i A_j\}_{seq}$, encoded via the matrix $X$, that achieve the quantum bound are unique and the uniqueness is robust to small deviations from the maximal value. Since the underlying single system isometries do not allow us to self-test the state or the measurement POVMs, we phrase our results in terms of the uniqueness of the set of correlations that maximally violate the inequality. Such a unique set of correlations for $\{A_i A_j\}_{seq}$ uniquely restricts the relationship between the measurements performed on the quantum system. For example, in the single qubit case with $N=3$, the successive measurement vectors on the Bloch sphere must be at $60^{\circ}$ to achieve maximal violation. This physical restriction carries over to higher values of $N$ as well. We contrast our certification scheme with the entanglement based self-testing formalism in Table \ref{table_1}. \\

\begin{table}
\centering
\begin{tabular}{|p{0.5\textwidth}|p{0.5\textwidth}|}
\hline  
    Certification of temporal correlations & Entanglement based self-testing\\
    \hline
\begin{itemize}[leftmargin=*]
  \item Based on non-classical temporal correlations
  \item Such correlations cannot be explained by models that obey the assumptions of Macro-realism
  \item Under the SDP formulation, we show uniqueness of optimizer matrix for the $N$-cycle temporal inequality; however, this does not self test the underlying state or measurements due to single system isometries involved. Instead, we obtain a unique set of correlations that restricts the relationship between measurement operators.
  \item Further, considering the generalized scenario with quantum channels we show that maximal violation of 3-cycle inequality certifies channel type and its Kraus rank. In the PDM formalism, the channels achieving this maximal violation correspond to the ``pseudo-Bell" states \cite{temporal_teleport}.
  \end{itemize} &
  \begin{itemize}[leftmargin=*]
  \item Based on non-local spatial correlations
  \item Such correlations cannot be explained by Local Hidden Variable models
  \item Maximal violation of Bell inequalities self tests underlying entangled state and measurements \cite{supic_self} with robust extensions known in certain scenarios \cite{RST_steering}.
  \item SDP formulations can also be constructed for Bell scenarios \cite{bharti2022graph}. Uniqueness of the optimizer matrix, if it holds as is the case for the CHSH inequality, corresponds to a self test of the underlying state (Bell state for the CHSH case) and measurements.
\end{itemize}\\
\hline
\end{tabular}
\caption{Comparison of our temporal correlations based certification scheme with entanglement based self testing scheme for state and measurement certification.}
\label{table_1}
\end{table}

\textit{Pseudo-density matrix as a unified formalism:} Consider the spatial analogue of inequality \eqref{seq_ineq_PDM} with $N=4$ where the measurement events are performed on spatially separated qubits (i.e. the CHSH inequality). We can infer that the PDMs maximally violating the inequality (with $N=4$) are the familiar maximally entangled Bell states which are valid density matrices. This example highlights an important feature of the PDM formalism as the one that can unite both spatial and temporal correlation based inequalities. If we restrict the PDM to be a valid density matrix by constraining it to be positive semi-definite (which in turn enforces the measurements to be performed on spatially separated qubits as described in Section \ref{PDM_appendix}), we recover results from Bell scenarios. However, lifting the constraint of positive semi-definiteness of PDMs also allows for dealing with sequential inequalities of the generalized scenario. \\

\textit{Local isometry in time:} As a follow up to the proof of Lemma \ref{lemma_4}, one can note that the tensor structure introduced between observables $A_m$ and $A_n$ allows for certain local isometries to act in the temporal sense. See that,

 \begin{equation} \label{eq:local_iso}
    \begin{aligned}
 \mathrm{Tr}\left[(A_{m}\otimes A_{n})R_{AB}\right] &= \mathrm{Tr}\left[(A_{m}\otimes A_{n})(\mathcal{I}_A \otimes \mathcal{E}_{A|B}) \Big\{ \rho_A\otimes \frac{I}{2},\; \frac{1}{2} \sum_{i=0}^{3} \sigma_i \otimes \sigma_i \Big\} \right] \\
        &= \mathrm{Tr} \left[(A_{m}\otimes A_{n}) \frac{1}{2} \sum_{i=0}^3 (\{\rho_A, \sigma_i\} \otimes \mathcal{E}_{A|B} (\sigma_i) \right] \\
        &= \frac{1}{2} \sum_{i=0}^3 \mathrm{Tr} \left[ A_m (\rho_A \sigma_i + \sigma_i \rho_A) \right] \mathrm{Tr} \left[ A_n \mathcal{E}_{A|B} (\sigma_i)  \right] \\
        &= \frac{1}{2} \sum_{i=0}^3 \mathrm{Tr} \left[ V A_m V^{\dagger} ( V\rho_AV^{\dagger} V\sigma_i V^{\dagger} + V\sigma_i V^{\dagger} V\rho_A V^{\dagger}) \right] \mathrm{Tr} \left[ UA_n U^{\dagger} \tilde{\mathcal{E}}_{A|B} (V\sigma_i V^{\dagger})  \right] \\
        &= \frac{1}{2} \sum_{i=0}^3 \mathrm{Tr} \left[ \tilde{A}_m ( \tilde{\rho}_A \sigma_i + \sigma_i \tilde{\rho}_A) \right] \mathrm{Tr} \left[ \tilde{A}_n \tilde{\mathcal{E}}_{A|B} (\sigma_i)  \right] \\
        &= \mathrm{Tr}\left[(VA_{m}V^{\dagger} \otimes U A_{n} U^{\dagger})(\mathcal{I}_A \otimes \tilde{\mathcal{E}}_{A|B}) \Big\{ V \rho_A V^{\dagger} \otimes \frac{I}{2},\; \frac{1}{2} \sum_{i=0}^{3} \sigma_i \otimes \sigma_i \Big\} \right] \\
\end{aligned}
\end{equation}
where $\tilde{\mathcal{E}}_{A|B}$ has transformed Kraus operators as $\{K_i\} \longrightarrow \{ UK_i V^{\dagger} \}$, $U$ is a unitary operator acting at time point $t_B$ and $V$, acting at $t_A$, is restricted to $V \in \{ X,Y,Z\}$ such that applying the transformation $V \sigma_i V^{\dagger}$ conserves the elements of the set $\sigma_i \in \{ \mathbb{I}, X, Y, Z \}$. Note that we make this particular choice for $V$ just to highlight the local isometry in time, there exist other choices for $V$ such as the Clifford group that would also conserve the set of Pauli matrices. Thus, the transformation of the channel $\mathcal{E}_{A|B}$'s Kraus operators  along with corresponding rotation of observables at times $t_A$ and $t_B$ does not change the correlation value. This local isometry in time is reminiscent of the local isometries present in the scenario with spatially separated entangled states and observables. \\

\textit{Applications:} We start by noting as the foremost application that obtaining an $(\epsilon,1)$ certificate on the full set of correlations only requires experimentally obtaining a subset of the correlations (for which $\lambda_{ij} \neq 0$). For example in the case of N-cycle objective function \eqref{eq:sn}, one only needs to measure the cyclic correlations $\{\langle A_{i} A_{i+1} \rangle \}$ scaling as $\mathcal{O}(N)$. Near optimal value of this subset gives a  robust guarantee on all pairs of correlations  $\{\langle A_{i} A_{j} \rangle \} \, \forall i,j$ which scales as $\mathcal{O}(N^2)$ (see Fig.~\ref{fig_thm_app}).

\begin{figure}[h]
  \includegraphics[width=0.85\columnwidth]{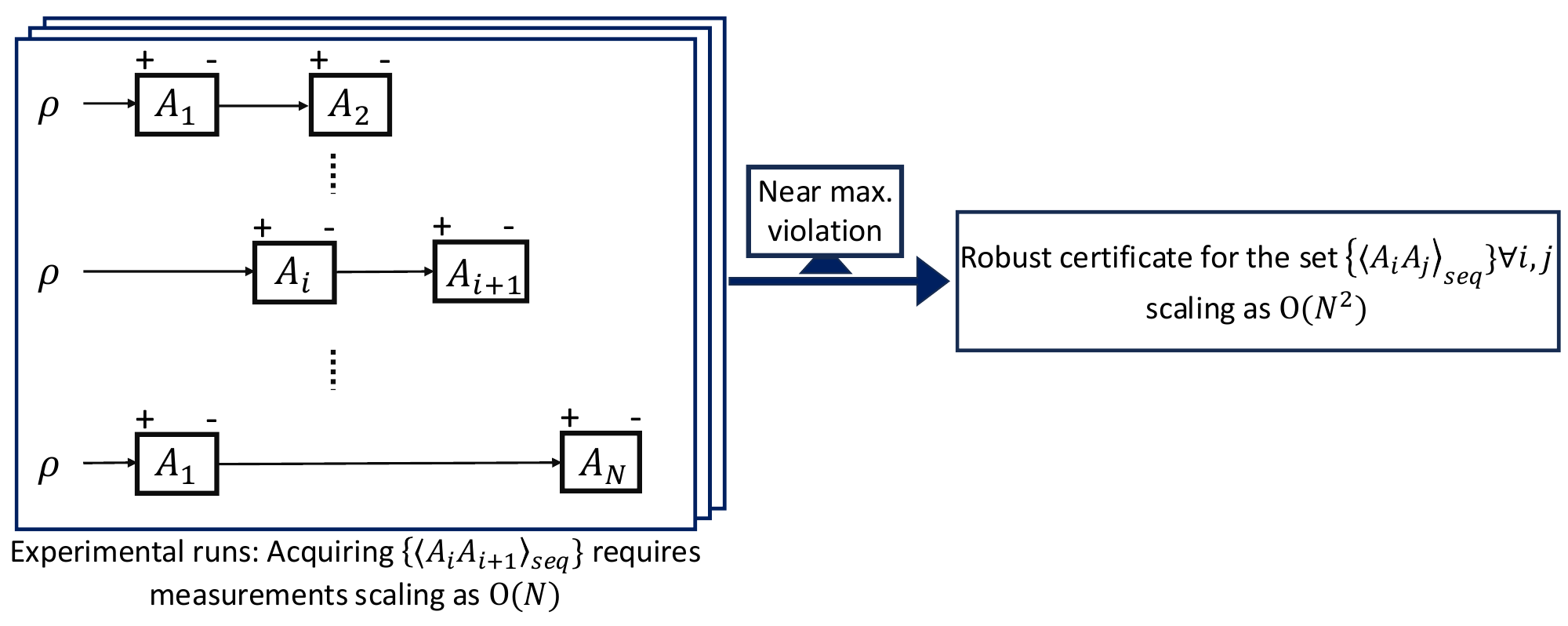}
  \caption{Our theorems allow an experimenter to obtain guarantees on all possible correlations of sequential measurements by measuring only a subset of sequential correlations, given that the measured correlations obtain near maximal violation.}
  \label{fig_thm_app}
\end{figure}

The full set of correlations thus obtained can be utilised for various applications such as dimension witnesses \cite{temporal_dim_witness,hoffmann2018structure}; under assumption of dimension of state being 2 they can be used for lower bounding purity \cite{spee2020certifying}. However, both these applications have been proposed for generalised measurements. Other applications which could be formulated for projective measurements are measurement device certification in semi-device independent scenario \cite{mohan2019sequential,das2022temporal}. In particular we can utilise (and also implement experimentally) the temporal inequalities proposed in \cite{das2022temporal} to certify measurement devices under the assumption that the initially prepared state is maximally mixed.

Coming to the channel certification result, previous schemes \cite{channel2018certifying} for quantum channel certification utilized entanglement as a resource where the channel was certified via its action on one half of a Bell state. Our scheme circumvents this requirement by utilizing sequential correlations on a single system. Secondly, our result can also be utilised for certifying building blocks for quantum circuits and  computing architectures. Since such architectures have been proposed to incorporate qubits, we can work under the assumption of 2-dim. states and utilise Theorem \ref{channel_cert_thm} to certify the unitarity of single qubit gates modeled as channels (this requires establishing the robustness properties of channel certification). Further, certain proposals \cite{qpt_baldwin} for quantum process tomography work under the assumption of unitarity of the underlying CPTP process. Working under this assumption they show improvements in number of elements required for characterizing the unitary channel. Our channel certification scheme could be used as a pre-cursor to guarantee the unitarity of the quantum channel.

\section{Discussion}
\label{sec:conclusion}
In this work we study correlations arising in 2-outcome sequential measurement scenarios. Unlike spatial correlations obtained in entanglement based scenarios, such sequential/temporal correlations have not seen many applications in the literature. Motivated by the SDP based formulations for bounding temporal correlations, we showed uniqueness of the optimizer matrix as well as establish robustness of this uniqueness property. Since the optimizer matrix is made up of sequential correlations between all pairs of measurements $\{ A_i \}$, it follows as a consequence of our results that near maximal violation of temporal correlation based inequalities can be used to obtain robust certificates of the set of sequential correlations $\{ \langle A_i, A_j \rangle_{seq} \}$. As an application of this result we highlight that any temporal inequality can be mapped to a cost function in the SDP formulation. Thus, near maximal violation of such an inequality requires experimentally measuring a subset of correlations usually scaling as $\mathcal{O}(N)$ where $N$ is the number of distinct 2-outcome measurements. However, by virtue of our result the certification of the full optimizer matrix certifies the full set of correlations scaling as $\mathcal{O}(N^2)$.

Next, we considered the generalized scenario with quantum channels acting in between the sequence of measurements to transfer the quantum state between 2 parties. We connected sequential correlations obtained in this scenario with Pseudo-density matrices (PDMs) via Lemma \ref{lemma_4}. This connection allows us to establish analogies between the structure of spatial and temporal inequalities along with local isometries involved. Further, we show that maximal violation of sequential inequality \eqref{eq:sn} $S_N \leq N-2 $ with  $N=3$ in the generalised scenario implies unitarity of the channel (Pauli channel with Kraus rank 1). This result could prove to be useful in implementing quantum process tomography protocols \cite{gutoski2014process, qpt_baldwin} which make the assumption of working with unitary channels. 

In the present work, our channel certification result is exact since it requires exact maximal violation to certify the channel. However, to be experimentally testable we require a robust version of our result i.e. future work could explore if near-maximal violation ($\mathcal{O}(\epsilon)$) of the inequality \eqref{eq:3cycle} certifies the near-unitarity of the Pauli channel with small error ($\mathcal{O}(\epsilon^{1/n})$). Further, we focus on 2-dim. states and measurements in our channel certification result owing to the fact that all possible quantum channels in this scenario admit a parametrization as a convex closure in the Pauli basis \eqref{eq:chn_convex}. It would be interesting to see if such parametrizations could be found for the general case of d-dim. states and measurements. This would pave the way for certifying quantum channels in the general case.

\section*{Acknowledgements}
The authors thank Otfried Gühne,  Valerio Scarani and Adan Cabello for helpful discussions. We are grateful to the National Research Foundation and the Ministry of Education, Singapore for financial support.  K.B.~acknowledges funding by AFOSR, DoE QSA, NSF QLCI (award No.~OMA-2120757), DoE ASCR Accelerated Research in Quantum Computing
program (award No.~DE-SC0020312), NSF PFCQC program, the DoE ASCR
Quantum Testbed Pathfinder program (award No.~DE-SC0019040),
U.S.~Department of Energy Award No.~DE-SC0019449, ARO MURI, AFOSR
MURI, DARPA SAVaNT ADVENT and A*STAR C230917003.

\section*{Data Availability}
Data sharing is not applicable to this article as no new data were created or analyzed in this study.

\section*{Author Declarations}
\subsection{Conflicts of Interest}
The authors have no conflicts to disclose.

\bibliography{temporal_cert.bib}

\end{document}